\documentclass[11pt]{article}

\usepackage[a4paper,margin=1in]{geometry}
\usepackage{amsmath,amssymb,amsfonts,amsthm,mathtools}
\usepackage[ruled,vlined,linesnumbered]{algorithm2e}
\usepackage{booktabs}
\usepackage{multirow}
\usepackage{graphicx}
\usepackage{xcolor}
\usepackage[hidelinks,hypertexnames=false]{hyperref}
\usepackage{authblk}
\usepackage{enumitem}
\usepackage{tabularx}
\usepackage{array}
\usepackage{ragged2e}
\usepackage{placeins}
\usepackage{float}

\newtheorem{proposition}{Proposition}

\newcommand{\out}{\mathrm{out}}
\newcommand{\inn}{\mathrm{in}}

\title{
Sender--Receiver Community Detection in Directed Networks\\
via Node-Role-Constrained Edge Clustering
}

\author{
Duy Hieu DO\\
Institute of Mathematics, Vietnam Academy of Science and Technology
}
\date{}

\begin{document}
\maketitle

\begin{center}
\texttt{ddhieu@math.ac.vn}
\end{center}

\begin{abstract}
Directed community detection is challenging because edge directions encode
asymmetric source--target relations. Most directed modularity and
random-walk methods assign one label to each vertex, whereas recent
bimodularity-based methods cluster directed edges more freely. We propose
TT-SR, a Two-Tier Sender--Receiver framework that lies between these two
viewpoints. Each vertex is assigned a sender role and a receiver role, and
each directed edge receives the type induced by the sender role of its
source and the receiver role of its target. Thus, TT-SR is more expressive
than one-label vertex clustering while remaining more interpretable than
unrestricted edge clustering. The method generates candidate
sender--receiver assignments from count-residual, stationary-flow,
degree-corrected, and order-score views. The candidates are refined by local
role updates and selected by a two-tier rule: a degree-corrected profile
score provides the primary structural criterion, while Bernoulli density and
order-flow scores are used only as secondary ranking signals. We justify the
main spectral views through sender--receiver modularity relaxations and
interpret the degree-corrected score as a likelihood-based residual
comparison. Experiments on pathway-type, co-block, and ordered-flow
synthetic benchmarks show that TT-SR achieves the strongest or essentially
tied strongest edge-community recovery across three scale settings. The
gains are most pronounced on degree-corrected co-block and ordered-flow
graphs. Real-network diagnostics further indicate that TT-SR aligns well
with Email-Eu-core metadata and extracts strong sender--receiver
bicommunity summaries on unlabeled directed networks.
\end{abstract}

\noindent\textbf{Keywords:} directed networks; community detection; edge clustering; sender--receiver roles; bimodularity; degree correction; stochastic block model; spectral initialization; random walks.

\section{Introduction}

Community detection is one of the central problems in network science. 
In undirected graphs, a community is commonly understood as a set of vertices with more internal edges than expected under a suitable null model. 
This idea has led to a large family of methods, including edge-betweenness and modularity-based algorithms \cite{GirvanNewman2002,NewmanGirvan2004,Newman2006Modularity}, multilevel modularity optimization such as Louvain and Leiden \cite{Blondel2008Louvain,Traag2019Leiden}, spectral clustering \cite{NgJordanWeiss2002,VonLuxburg2007}, random-walk methods \cite{PonsLatapy2005,RosvallBergstrom2008,RosvallAxelssonBergstrom2009,Delvenne2010}, and stochastic-block-model approaches \cite{Holland1983SBM,KarrerNewman2011}. 
Comprehensive surveys can be found in \cite{Fortunato2010,FortunatoHric2016}.

Directed networks require a more careful formulation. 
A directed edge \(i\to j\) does not merely indicate that \(i\) and \(j\) are similar or connected. 
It indicates that \(i\) acts as a source and \(j\) acts as a target. 
This distinction is essential in many applications. 
In citation networks, a paper may cite one research area while being cited by another. 
In trade networks, a country may export to one group of countries while importing from another. 
In communication, information-flow, and neuronal networks, one group may primarily transmit information while another primarily receives it. 
Thus, directed community structure is often a structure of flows from sender roles to receiver roles, not simply a partition of vertices.

A standard way to extend modularity to directed graphs is to compare the observed edge \(A_{ij}\) with the directed configuration-model expectation
\[
\frac{d_i^{\out}d_j^{\inn}}{m},
\]
where \(d_i^{\out}\) is the out-degree of vertex \(i\), \(d_j^{\inn}\) is the in-degree of vertex \(j\), and \(m\) is the total edge weight. 
Leicht and Newman introduced a directed modularity based on this out-degree/in-degree null model \cite{LeichtNewman2008}. 
This correction is essential because directed networks may have very different sending and receiving activity patterns. 
However, many classical directed modularity formulations, including the Leicht--Newman formulation, still assign one label \(g_i\) to each vertex. 
Such methods therefore do not directly model the possibility that vertex \(i\) may belong to one group as a sender and another group as a receiver.

Recently, Cionca, Chan, and Van De Ville revisited directed community detection through the notion of \emph{bimodularity} \cite{Cionca2025Bimodularity}. 
Their approach distinguishes sending and receiving communities, constructs the directed modularity matrix
\[
B_A
=
A-\frac{d^{\out}(d^{\inn})^\top}{m},
\]
uses its singular value decomposition, and clusters edge features formed from source-side and target-side singular-vector coordinates. 
This is an important conceptual step because it treats directed communities as source--target pathways rather than as ordinary undirected-like vertex groups.

The present paper is closest in spirit to this line of work, but it imposes a different structural constraint. 
Instead of clustering edges as independent objects in an edge-feature space, we require every edge label to be induced by two reusable node roles: the sender role of its source and the receiver role of its target. 
Thus, an edge \(i\to j\) is assigned the type
\[
(z_i^{\out},z_j^{\inn}).
\]
This node-role-constrained formulation is more expressive than one-label vertex clustering, because a vertex may send and receive according to different roles. 
At the same time, it is more interpretable than unrestricted edge clustering, because the edge communities are summarized by a compact sender--receiver role system at the vertex level.

The theoretical starting point of our framework is the random-walk modularity for directed graphs introduced in \cite{DoPhanDang2023}. 
Let \(P\) be an ergodic transition matrix and let \(\phi\) be its stationary distribution. 
In source--target notation, the stationary flow from \(u\) to \(v\) is \(\phi_uP_{uv}\), while the independent stationary-flow null model is \(\phi_u\phi_v\). 
For a one-label partition \(C_1,\ldots,C_K\), the stationary-distribution modularity is
\[
Q_{\rm RW}
=
\sum_{u,v}
\left(\phi_uP_{uv}-\phi_u\phi_v\right)
\delta_{C_uC_v}.
\]
The residual term is already directed, but the compatibility condition \(\delta_{C_uC_v}\) is still one-role: the source and target vertices are compared through the same community label.

TT-SR keeps the directed residual but changes the role model. 
Each vertex \(i\) receives a sender role \(z_i^{\out}\) and a receiver role \(z_i^{\inn}\). 
More generally, we introduce a sender--receiver compatibility matrix \(\Omega\) and define
\[
Q_{\rm SR}(\Omega)
=
\sum_{u,v}
\left(\phi_uP_{uv}-\phi_u\phi_v\right)
\Omega_{z_u^{\out},z_v^{\inn}}.
\]
The diagonal choice \(\Omega_{ab}=\mathbf 1\{a=b\}\) gives the direct two-role analogue
\[
Q_{\rm SR}^{\rm diag}
=
\sum_{u,v}
\left(\phi_uP_{uv}-\phi_u\phi_v\right)
\delta_{z_u^{\out},z_v^{\inn}}.
\]
The full edge-community representation is then obtained by keeping the ordered role pair
\[
i\to j
\quad\mapsto\quad
(z_i^{\out},z_j^{\inn}).
\]
Thus, the conceptual bridge is not merely from \(Q_{\rm RW}\) to another modularity score, but from one-label random-walk modularity to node-role-constrained directed edge communities.

This role-pair representation places the proposed model between one-label vertex clustering and unrestricted edge clustering. 
If one imposes
\[
z_i^{\out}=z_i^{\inn}
\quad\text{for all } i,
\]
then the model reduces to a one-label-per-node partition. 
If the two roles are allowed to differ, the model can represent co-block structures, source--target pathways, cyclic block flows, ordered directed flows, and bipartite-like directed structures. 
At the same time, the model remains more structured than unrestricted edge clustering because all edge labels are induced by two node-level role assignments.

The main algorithmic challenge is that directed networks may reveal their structure through different signals. 
Some networks are dominated by raw density. 
Some require degree correction because high-activity vertices dominate the edge counts. 
Some are better described by stationary random-walk flow. 
Some contain a global ordering in which edges tend to point from earlier groups to later groups. 
A less robust strategy would be to first identify the type of the input graph and then choose a specialized algorithm. 
Our strategy is different: we use one sender--receiver model, generate many candidate partitions from multiple generic directed views, refine them by common objectives, and select the final solution by a two-tier internal criterion.

Thus, the main methodological distinction of TT-SR is its role-constrained view of directed edge communities: edge labels are not free clusters, but are induced by reusable sender and receiver roles.

This paper also continues our previous line of work on random-walk and flow-based community detection, including random-walk distances, Louvain-type refinements, dynamic community detection, overlapping communities, and stationary-distribution modularity for directed graphs \cite{DoPhan2022RIVF,DoPhan2025JOCO,DoNguyenPhan2025RIVF,DoPhan2025ACS,DoPhanDang2023}. 
In contrast to these earlier one-label or vertex-centered settings, the present work moves the random-walk viewpoint to a sender--receiver edge-community setting.

The contributions of this paper are as follows.

\begin{itemize}[leftmargin=2em]
    \item We introduce a node-role-constrained formulation of directed edge community detection. 
    Each vertex is assigned a sender role and a receiver role, and every directed edge \(i\to j\) receives the induced type
    \[
    (z_i^{\out},z_j^{\inn}).
    \]
    This provides a middle ground between one-label vertex clustering and unrestricted edge clustering.

    \item We build on the stationary-distribution modularity \(Q_{\rm RW}\) for directed graphs proposed in \cite{DoPhanDang2023} and derive a sender--receiver modularity family
    \[
    Q_{\rm SR}(\Omega)
    =
    \sum_{u,v}
    \left(\phi_uP_{uv}-\phi_u\phi_v\right)
    \Omega_{z_u^{\out},z_v^{\inn}}.
    \]
    The diagonal case recovers a direct two-role analogue of \(Q_{\rm RW}\), while the full role-pair representation supports off-diagonal directed pathways.

\item We show that the sender and receiver spectral coordinates used by TT-SR arise from continuous relaxations of sender--receiver modularity objectives. 
For a directed residual matrix \(B\) and spectral dimension \(r\), the relaxed problem
\[
\max_{X^\top X=I_r,\;Y^\top Y=I_r}\operatorname{tr}(X^\top B Y)
\]
is solved by the leading \(r\) left and right singular vectors of \(B\). 
Thus, the source-side and target-side coordinates are not ad hoc embeddings, but relaxed optimizers of a directed sender--receiver modularity objective.

\item We introduce a multi-view initialization strategy based on count residuals, stationary-flow residuals, degree-corrected residuals, and order-score segmentations. 
The count-residual and stationary-flow views follow from ordinary sender--receiver spectral relaxations, while the degree-corrected view follows from a degree-weighted relaxation that normalizes outgoing and incoming activity separately.

\item We introduce a degree-corrected sender--receiver profile score whose unsmoothed form is an aggregated Poisson likelihood-ratio statistic under a directed configuration null model. 
We further relate this profile score to sender--receiver residual modularity by showing that its local quadratic approximation is a degree-normalized residual energy, which is, up to a constant factor, the square of the strongest normalized linear sender--receiver residual contrast. 
Thus, the degree-corrected score is not an unrelated selection heuristic, but a likelihood-based nonlinear counterpart of the residual modularity principle.

\item We propose local refinement procedures for sender and receiver roles under degree-corrected and Bernoulli co-block objectives.

\item We formulate the two-tier selection rule as a constrained
candidate-selection principle. 
The degree-corrected profile score first defines a primary support set of
structurally reliable candidates. 
Inside this support set, a rank-based secondary score combines
degree-corrected, Bernoulli density, and order-flow ranks. 
Thus, auxiliary criteria are used only as secondary evidence among
degree-corrected-eligible candidates, rather than as unconstrained
competing objectives.
    \item We design a three-scale synthetic experimental protocol covering pathway-type directed communities, sender--receiver co-block communities, and ordered-flow directed communities. 
    The experiments show that TT-SR remains accurate as both the graph size and the number of planted groups increase.
\end{itemize}

\section{Background and Positioning}

\subsection{Directed modularity, flow modularity, and degree correction}

Modularity compares observed edge density with a suitable null model
\cite{NewmanGirvan2004,Newman2006Modularity}. In directed graphs, the null
model must preserve outgoing and incoming activity separately. Let
\[
d_i^{\mathrm{out}}=\sum_j A_{ij},
\qquad
d_j^{\mathrm{in}}=\sum_i A_{ij},
\qquad
m=\sum_{i,j}A_{ij}.
\]
The directed configuration-model expectation is
\[
\mathbb E_0[A_{ij}]
=
\frac{d_i^{\mathrm{out}}d_j^{\mathrm{in}}}{m},
\]
which gives the directed modularity matrix
\[
B_A
=
A-\frac{d^{\mathrm{out}}(d^{\mathrm{in}})^\top}{m}.
\]
The corresponding one-label directed modularity objective is
\[
Q_{\mathrm{dir}}
=
\frac{1}{m}
\sum_{i,j}
\left(
A_{ij}
-
\frac{d_i^{\mathrm{out}}d_j^{\mathrm{in}}}{m}
\right)
\mathbf 1\{g_i=g_j\}.
\]
Although this objective uses a directed null model, it still assigns a
single label \(g_i\) to each vertex. It therefore cannot directly represent
vertices that play different roles as senders and receivers.

A related directed modularity is based on stationary random-walk flow
\cite{DoPhanDang2023}. If \(P\) is an ergodic transition matrix and
\(\phi\) is its stationary distribution, then the observed stationary flow
from \(u\) to \(v\) is \(\phi_uP_{uv}\), while the independent-flow null
model is \(\phi_u\phi_v\). For a one-label partition, this gives
\[
Q_{\rm RW}
=
\sum_{u,v}
\left(
\phi_uP_{uv}-\phi_u\phi_v
\right)
\delta_{C_uC_v}.
\]
This preserves edge direction, but it still compares the source and target
vertices through the same community label. The present paper keeps the
directed residual viewpoint but changes the role model by assigning each
vertex a sender role and a receiver role.

Degree correction is also essential. The degree-corrected stochastic block
model shows that ignoring heterogeneous node activity can distort community
detection \cite{KarrerNewman2011}. In directed networks, heterogeneity
appears separately in outgoing and incoming activity. This motivates the
degree-corrected sender--receiver residuals and profile score used in TT-SR.

\subsection{Relation to bimodularity, co-clustering, and flow methods}

Bimodularity-based directed community detection takes an explicitly
source--target viewpoint. Cionca, Chan, and Van De Ville define directed
communities as mappings from sending communities to receiving communities
and use the singular value decomposition of the directed modularity matrix
\[
B_A=U\Sigma V^\top
\]
to construct source-side and target-side edge features
\cite{Cionca2025Bimodularity}. These edge features are then clustered
directly, producing unrestricted edge communities.

TT-SR is closest in spirit to this bimodularity viewpoint, but it imposes a
node-role constraint on the final representation. Instead of clustering
edges as independent objects, TT-SR assigns two reusable roles to each
vertex and induces the edge label by
\[
i\to j
\quad\mapsto\quad
(z_i^{\mathrm{out}},z_j^{\mathrm{in}}).
\]
Thus, TT-SR lies between one-label vertex clustering and unrestricted edge
clustering: it is more expressive than assigning one label to each vertex,
but more interpretable than free edge clustering.

The distinction between source-side and target-side structure also appears
in directed co-clustering and spectral methods. DiSIM detects asymmetric
co-clusters under the directed stochastic co-blockmodel
\cite{RoheQinYu2016}, while D-SCORE-type methods use degree-corrected
spectral coordinates for directed community detection
\cite{WangLiangJi2020}. Random-walk methods provide another related
viewpoint, including Walktrap, Infomap, Markov stability, and
stationary-distribution approaches
\cite{PonsLatapy2005,RosvallBergstrom2008,RosvallAxelssonBergstrom2009,Delvenne2010,DoPhanDang2023}.
In TT-SR, these ideas are not used as separate competing models; rather,
count residuals, stationary-flow residuals, degree-corrected residuals, and
order scores are used as complementary candidate-generating views within a
single sender--receiver framework.

\section{Sender--Receiver Theory and Scores}
\label{sec:sr-theory}

\subsection{Sender--receiver modularity and edge types}

Let \(G=(V,E)\) be a directed graph with adjacency matrix
\(A\in\mathbb R^{n\times n}\), where \(A_{ij}\ge0\) is the weight of the
directed edge \(i\to j\). Let \(P\) be an ergodic row-stochastic transition
matrix and let \(\phi\) be its stationary distribution,
\[
\phi^\top P=\phi^\top,
\qquad
\sum_i\phi_i=1.
\]
The stationary flow from \(u\) to \(v\) is \(\phi_uP_{uv}\), while the
independent stationary-flow null model is \(\phi_u\phi_v\). Thus the
random-walk residual matrix is
\[
B^{\rm rw}
=
\Phi P-\phi\phi^\top,
\qquad
\Phi=\operatorname{diag}(\phi).
\]
The stationary-distribution modularity of Dang, Do, and Phan
\cite{DoPhanDang2023} can be written as
\[
Q_{\rm RW}
=
\sum_{u,v}
\left(
\phi_uP_{uv}-\phi_u\phi_v
\right)
\delta_{C_uC_v}
=
\operatorname{tr}(S^\top B^{\rm rw}S),
\]
where \(S\) is the membership matrix of the one-label partition. This score
preserves edge direction, but it still compares the source and target
vertices through the same community label.

To model asymmetric source--target behavior, TT-SR assigns two labels to
each vertex:
\[
z_i^{\mathrm{out}}\in\{1,\ldots,K_{\mathrm{out}}\},
\qquad
z_i^{\mathrm{in}}\in\{1,\ldots,K_{\mathrm{in}}\}.
\]
The label \(z_i^{\mathrm{out}}\) describes how vertex \(i\) sends flow, while
\(z_i^{\mathrm{in}}\) describes how it receives flow. Let
\(S^{\mathrm{out}}\) and \(S^{\mathrm{in}}\) be the corresponding sender and
receiver membership matrices. The aggregated sender--receiver random-walk
residual is
\[
H^{\rm rw}
=
(S^{\mathrm{out}})^\top B^{\rm rw}S^{\mathrm{in}}.
\]
For a compatibility matrix
\(\Omega\in\mathbb R^{K_{\mathrm{out}}\times K_{\mathrm{in}}}\), define
\[
Q_{\rm SR}(\Omega)
=
\left\langle
(S^{\mathrm{out}})^\top B^{\rm rw}S^{\mathrm{in}},
\Omega
\right\rangle
=
\sum_{u,v}
B^{\rm rw}_{uv}
\Omega_{z_u^{\mathrm{out}},z_v^{\mathrm{in}}}.
\]
The diagonal choice \(K_{\mathrm{out}}=K_{\mathrm{in}}=K\) and
\(\Omega=I_K\) gives
\[
Q_{\rm SR}^{\rm diag}
=
\sum_{u,v}
\left(
\phi_uP_{uv}-\phi_u\phi_v
\right)
\delta_{z_u^{\mathrm{out}},z_v^{\mathrm{in}}}
=
\operatorname{tr}\left(
(S^{\mathrm{out}})^\top B^{\rm rw}S^{\mathrm{in}}
\right).
\]
If \(z_i^{\mathrm{out}}=z_i^{\mathrm{in}}=C_i\) for all vertices, then
\(Q_{\rm SR}^{\rm diag}\) reduces to the one-label modularity \(Q_{\rm RW}\).
Thus the sender--receiver model contains the usual one-label model as a
special case.

For directed edge community detection, however, the diagonal compatibility
is still too restrictive because meaningful directed pathways may be
off-diagonal, such as \(1\to2\), \(2\to3\), or \(3\to1\). TT-SR therefore
keeps the full ordered role pair as the edge type:
\[
y_{ij}
=
(z_i^{\mathrm{out}},z_j^{\mathrm{in}}).
\]
If a single integer label is required, we encode it as
\[
\widehat y_{ij}
=
K_{\mathrm{in}}(z_i^{\mathrm{out}}-1)+z_j^{\mathrm{in}}.
\]
This role-pair representation is more expressive than one-label vertex
clustering, while remaining more interpretable than unrestricted edge
clustering because all edge labels are induced by reusable node roles.

\subsection{Spectral relaxation}

The sender--receiver modularity objective is discrete because the membership
matrices are binary. We now show that the left and right singular-vector
coordinates used by TT-SR arise from a continuous relaxation.

Let \(B\in\mathbb R^{n\times n}\) be a centered directed residual matrix.
For example,
\[
B=B^{\rm rw}=\Phi P-\phi\phi^\top
\]
gives the stationary-flow view, while
\[
B=B_A=A-\frac{d^{\mathrm{out}}(d^{\mathrm{in}})^\top}{m}
\]
gives the count-residual view. Let \(S^{\mathrm{out}}\) and
\(S^{\mathrm{in}}\) be sender and receiver membership matrices with
nonempty groups, and define the normalized matrices
\[
F^{\mathrm{out}}
=
S^{\mathrm{out}}
\left((S^{\mathrm{out}})^\top S^{\mathrm{out}}\right)^{-1/2},
\qquad
F^{\mathrm{in}}
=
S^{\mathrm{in}}
\left((S^{\mathrm{in}})^\top S^{\mathrm{in}}\right)^{-1/2}.
\]
Then
\[
(F^{\mathrm{out}})^\top F^{\mathrm{out}}=I,
\qquad
(F^{\mathrm{in}})^\top F^{\mathrm{in}}=I.
\]
A normalized sender--receiver modularity relaxation is
\[
\mathcal Q_B
=
\operatorname{tr}
\left(
(F^{\mathrm{out}})^\top B F^{\mathrm{in}}
\right).
\]
Relaxing \(F^{\mathrm{out}}\) and \(F^{\mathrm{in}}\) to arbitrary
orthonormal matrices \(X,Y\in\mathbb R^{n\times r}\) gives
\[
\max_{X^\top X=I_r,\;Y^\top Y=I_r}
\operatorname{tr}(X^\top B Y).
\]

\begin{proposition}[SVD relaxation of sender--receiver modularity]
\label{prop:svd_relaxation}
Let \(B=U\Sigma V^\top\) be a singular value decomposition of \(B\), with
singular values
\[
\sigma_1\ge\sigma_2\ge\cdots\ge\sigma_n\ge0.
\]
Then, for any \(r\le n\),
\[
\max_{X^\top X=I_r,\;Y^\top Y=I_r}
\operatorname{tr}(X^\top B Y)
=
\sum_{\ell=1}^r \sigma_\ell.
\]
One maximizer is
\[
X^\star=U_r,
\qquad
Y^\star=V_r,
\]
where \(U_r\) and \(V_r\) contain the first \(r\) left and right singular
vectors of \(B\).
\end{proposition}

\begin{proof}
By the Ky Fan maximum principle, or equivalently von Neumann's trace
inequality,
\[
\operatorname{tr}(X^\top B Y)
\le
\sum_{\ell=1}^r \sigma_\ell
\]
for all matrices \(X,Y\in\mathbb R^{n\times r}\) satisfying
\(X^\top X=I_r\) and \(Y^\top Y=I_r\). Equality is attained by choosing
\(X=U_r\) and \(Y=V_r\). Hence the leading left and right singular vectors
solve the relaxed sender--receiver modularity problem.
\end{proof}

This proposition justifies the use of left singular vectors as sender-role
coordinates and right singular vectors as receiver-role coordinates. In
implementation, TT-SR uses singular-value weighted embeddings such as
\[
X^{\mathrm{out}}=U_r\Sigma_r,
\qquad
X^{\mathrm{in}}=V_r\Sigma_r,
\]
and then clusters their rows into sender and receiver roles.

The same relaxation also explains the degree-corrected residual view. Let
\[
W_{\mathrm{out}}
=
\operatorname{diag}(d^{\mathrm{out}}+\tau_{\mathrm{out}}),
\qquad
W_{\mathrm{in}}
=
\operatorname{diag}(d^{\mathrm{in}}+\tau_{\mathrm{in}}),
\]
where \(\tau_{\mathrm{out}},\tau_{\mathrm{in}}>0\). The degree-weighted
relaxation
\[
\max
\operatorname{tr}(X^\top B_A Y)
\quad
\text{subject to}
\quad
X^\top W_{\mathrm{out}}X=I_r,
\qquad
Y^\top W_{\mathrm{in}}Y=I_r
\]
is transformed, by
\[
\widetilde X=W_{\mathrm{out}}^{1/2}X,
\qquad
\widetilde Y=W_{\mathrm{in}}^{1/2}Y,
\]
into an ordinary orthonormal SVD relaxation for
\[
R
=
W_{\mathrm{out}}^{-1/2}
B_A
W_{\mathrm{in}}^{-1/2}.
\]
Equivalently,
\[
R_{ij}
=
\frac{
A_{ij}-d_i^{\mathrm{out}}d_j^{\mathrm{in}}/m
}{
\sqrt{(d_i^{\mathrm{out}}+\tau_{\mathrm{out}})
(d_j^{\mathrm{in}}+\tau_{\mathrm{in}})}
}.
\]
Thus the count-residual, stationary-flow residual, and degree-corrected
residual views used by TT-SR share the same spectral sender--receiver
relaxation principle.

\subsection{Degree-corrected profile score}

The spectral residual views generate candidate sender and receiver roles.
For final selection, TT-SR uses a degree-corrected profile score. Given
sender labels \(z^{\mathrm{out}}\) and receiver labels \(z^{\mathrm{in}}\),
define the observed block mass
\[
M_{ab}
=
\sum_{i,j}
A_{ij}
\mathbf 1
\{z_i^{\mathrm{out}}=a,\ z_j^{\mathrm{in}}=b\}.
\]
Let
\[
\kappa_a^{\mathrm{out}}
=
\sum_{i:z_i^{\mathrm{out}}=a}d_i^{\mathrm{out}},
\qquad
\kappa_b^{\mathrm{in}}
=
\sum_{j:z_j^{\mathrm{in}}=b}d_j^{\mathrm{in}}.
\]
Under the directed degree-corrected null model, the expected edge mass from
sender group \(a\) to receiver group \(b\) is
\[
E_{ab}
=
\frac{\kappa_a^{\mathrm{out}}\kappa_b^{\mathrm{in}}}{m}.
\]
The degree-corrected sender--receiver profile score is
\[
S_{\mathrm{DC}}(z^{\mathrm{out}},z^{\mathrm{in}})
=
\sum_{a,b}
M_{ab}
\log
\frac{M_{ab}+\varepsilon}{E_{ab}+\varepsilon}.
\]
The implemented score always uses the smoothed form with \(\varepsilon>0\).
Whenever the unsmoothed expressions are used for interpretation, they are
understood on block pairs with \(E_{ab}>0\), together with the usual
zero-count convention \(0\log(0/E_{ab})=0\).

A block pair is rewarded only when its observed edge mass is large relative
to what would be expected from the total outgoing activity of its sender
group and the total incoming activity of its receiver group. This makes
\(S_{\mathrm{DC}}\) more robust to heterogeneous out-degrees and in-degrees
than raw density scores.

The unsmoothed version of \(S_{\mathrm{DC}}\) has a simple likelihood
interpretation. If the aggregated counts are compared with the directed
configuration null
\[
M_{ab}\sim \operatorname{Poisson}(E_{ab}),
\]
then the log-likelihood ratio between the saturated block model with mean
\(M_{ab}\) and the null model is
\[
\sum_{a,b}
\left[
M_{ab}\log\frac{M_{ab}}{E_{ab}}
-
M_{ab}
+
E_{ab}
\right].
\]
Since
\[
\sum_{a,b}M_{ab}
=
\sum_{a,b}E_{ab}
=
m,
\]
the linear terms cancel, leaving
\[
\sum_{a,b}
M_{ab}\log\frac{M_{ab}}{E_{ab}},
\]
which is the unsmoothed form of \(S_{\mathrm{DC}}\).

Locally, if \(D_{ab}=M_{ab}-E_{ab}\) is small relative to \(E_{ab}\), then
\(M_{ab}=E_{ab}+D_{ab}\) and a Taylor expansion gives
\[
M_{ab}\log\frac{M_{ab}}{E_{ab}}
=
D_{ab}
+
\frac12\frac{D_{ab}^2}{E_{ab}}
+
O\left(\frac{|D_{ab}|^3}{E_{ab}^2}\right).
\]
Since
\[
\sum_{a,b}D_{ab}
=
\sum_{a,b}M_{ab}
-
\sum_{a,b}E_{ab}
=
0,
\]
the linear terms cancel after summing over all block pairs. Hence
\[
\sum_{a,b}
M_{ab}\log\frac{M_{ab}}{E_{ab}}
=
\frac12
\sum_{a,b}
\frac{D_{ab}^2}{E_{ab}}
+
O\left(
\sum_{a,b}
\frac{|D_{ab}|^3}{E_{ab}^2}
\right).
\]
Thus \(S_{\mathrm{DC}}\) can be viewed as a likelihood-based nonlinear
counterpart of the residual modularity principle: it profiles out the
block-pair intensities and measures degree-normalized sender--receiver
deviation from the directed configuration null.

This quadratic residual energy also has a simple linear-contrast
interpretation. For the normalized class
\[
\mathcal A
=
\left\{
\Omega:
\sum_{a,b}E_{ab}\Omega_{ab}^{2}\le 1
\right\},
\]
the Cauchy--Schwarz inequality gives
\[
\sum_{a,b}D_{ab}\Omega_{ab}
\le
\left(
\sum_{a,b}\frac{D_{ab}^{2}}{E_{ab}}
\right)^{1/2}.
\]
Equality is attained, when \(D\not\equiv 0\), by taking
\[
\Omega_{ab}
\propto
\frac{D_{ab}}{E_{ab}}.
\]
Consequently, the local quadratic approximation of
\(S_{\mathrm{DC}}\) is, up to the factor \(1/2\), the square of the
strongest normalized sender--receiver linear residual contrast.

In the synthetic experiments, the number of sender and receiver roles is
fixed to the planted value \(K\). In the real-network experiments, we report
diagnostics over prescribed \(K\)-grids. If candidates with different role
budgets are compared in a fully unsupervised setting, \(S_{\mathrm{DC}}\)
should be combined with a complexity penalty such as BIC, MDL, or
held-out likelihood \cite{Schwarz1978,Rissanen1978}.

\subsection{Auxiliary Bernoulli and order-flow scores}

The degree-corrected score is the primary structural criterion. TT-SR also
uses two auxiliary scores as secondary evidence during candidate selection.

For binary graphs, or for the binarized support of a weighted graph, define
\[
A^{(0)}_{ij}=\mathbf 1\{A_{ij}>0\}.
\]
Let
\[
M^{(0)}_{ab}
=
\sum_{i,j}
A^{(0)}_{ij}
\mathbf 1\{z_i^{\mathrm{out}}=a,\ z_j^{\mathrm{in}}=b\}.
\]
Let
\[
n_a^{\mathrm{out}}
=
|\{i:z_i^{\mathrm{out}}=a\}|,
\qquad
n_b^{\mathrm{in}}
=
|\{j:z_j^{\mathrm{in}}=b\}|.
\]
If self-loops are excluded, the number of possible directed pairs from
sender group \(a\) to receiver group \(b\) is
\[
N_{ab}
=
n_a^{\mathrm{out}}n_b^{\mathrm{in}}
-
|\{i:z_i^{\mathrm{out}}=a,\ z_i^{\mathrm{in}}=b\}|.
\]
With smoothing \(\alpha_{\rm B}>0\), define
\[
\widehat p_{ab}
=
\frac{M^{(0)}_{ab}+\alpha_{\rm B}}{N_{ab}+2\alpha_{\rm B}}.
\]
The Bernoulli score is
\[
S_{\rm B}
=
\sum_{a,b}
\left[
M^{(0)}_{ab}\log \widehat p_{ab}
+
(N_{ab}-M^{(0)}_{ab})\log(1-\widehat p_{ab})
\right].
\]
This score is useful for density-dominated co-block patterns, but it is
used only as a secondary score because it can be affected by degree
heterogeneity.

The order-flow score is designed for graphs with an approximate global
direction. Let \(s_i\) be a scalar order score for vertex \(i\), and let
\(\bar s_a^{\mathrm{out}}\) and \(\bar s_b^{\mathrm{in}}\) be the average
order scores of sender group \(a\) and receiver group \(b\). A block pair is
order-consistent if
\[
\bar s_b^{\mathrm{in}}
\ge
\bar s_a^{\mathrm{out}}-\delta,
\]
where \(\delta\ge0\) is a tolerance parameter. Let \(\mathcal F\) be the set
of order-consistent block pairs and let \(D_{ab}=M_{ab}-E_{ab}\). The
order-flow score is
\[
S_{\rm Ord}
=
\frac{1}{m}
\left[
\sum_{(a,b)\in\mathcal F} [D_{ab}]_+
-
\sum_{(a,b)\notin\mathcal F} [D_{ab}]_+
\right],
\qquad
[x]_+=\max\{x,0\}.
\]
It rewards positive degree-corrected excess flow that follows the learned
forward direction and penalizes positive excess flow that moves backward.

\section{The Proposed TT-SR Algorithm}
\label{sec:algorithm}

TT-SR, or \emph{Two-Tier Sender--Receiver}, detects directed edge
communities through reusable sender and receiver node roles. The algorithm
has three stages. It first generates candidate sender--receiver assignments
from several directed views, then refines each candidate by local role
updates, and finally selects the output through a degree-corrected gate
followed by secondary rank-based tie-breakers.

\subsection{Candidate generation}

TT-SR uses four complementary candidate-generating views. Each spectral view
produces a left embedding for sender roles and a right embedding for
receiver roles.

\paragraph{Count residual view.}
The first view uses the directed modularity residual
\[
B_A
=
A-\frac{d^{\mathrm{out}}(d^{\mathrm{in}})^\top}{m}.
\]
A truncated singular value decomposition
\[
B_A\approx U_r\Sigma_rV_r^\top
\]
gives sender and receiver embeddings
\[
X^{\mathrm{out}}=U_r\Sigma_r,
\qquad
X^{\mathrm{in}}=V_r\Sigma_r.
\]
Clustering the rows of these matrices gives an initial pair of sender and
receiver role assignments.

\paragraph{Stationary-flow residual view.}
Let \(P\) be a row-stochastic transition matrix derived from \(A\), with
dangling rows replaced by a uniform row. With teleportation parameter
\(\alpha_{\rm flow}\), define
\[
P_{\alpha_{\rm flow}}
=
(1-\alpha_{\rm flow})P
+
\alpha_{\rm flow}\frac{\mathbf 1\mathbf 1^\top}{n}.
\]
The teleportation term makes \(P_{\alpha_{\rm flow}}\) ergodic and ensures
that the stationary distribution is unique and numerically well defined.
Let \(\phi\) be the stationary distribution of \(P_{\alpha_{\rm flow}}\).
We then use \(\phi\) as a regularized source weight, while keeping the
original directed transition structure in the flow matrix
\[
F_{ij}
=
\phi_iP_{ij}.
\]
Thus this candidate view is a regularized stationary-flow view: teleportation
is used to stabilize the source weights, whereas the observed transition
matrix \(P\) is retained to preserve the original directionality of the
network. For \(\lambda\in[0,1]\), set
\[
G_\lambda
=
(1-\lambda)\frac{A}{m}
+
\lambda F.
\]
If \(r_i=\sum_j(G_\lambda)_{ij}\) and
\(c_j=\sum_i(G_\lambda)_{ij}\), the centered count-flow residual is
\[
B_\lambda
=
G_\lambda
-
\frac{rc^\top}{\sum_{i,j}(G_\lambda)_{ij}}.
\]
A truncated SVD of \(B_\lambda\) gives another pair of sender and receiver
embeddings.

\paragraph{Degree-corrected residual view.}
To reduce the influence of heterogeneous degrees, TT-SR also uses
\[
R_{ij}
=
\frac{
A_{ij}-d_i^{\mathrm{out}}d_j^{\mathrm{in}}/m
}{
\sqrt{
(d_i^{\mathrm{out}}+\tau_{\mathrm{out}})
(d_j^{\mathrm{in}}+\tau_{\mathrm{in}})
}
}.
\]
A truncated SVD of \(R\) gives degree-corrected sender and receiver
coordinates.

\paragraph{Order-score view.}
To capture global ordered flow, TT-SR computes a scalar score \(s_i\) by
solving
\[
\min_{s\in\mathbb R^n}
\sum_{i,j}W_{ij}(s_j-s_i-1)^2
+
\rho_{\rm ord}\sum_i s_i^2,
\]
where \(W\) is a nonnegative directed weight matrix such as \(A\) or
\(G_\lambda\). The normal equations have the form
\[
(L_W+\rho_{\rm ord}I)s=b_W,
\]
where
\[
L_W=
\operatorname{diag}\bigl((W+W^\top)\mathbf 1\bigr)
-
(W+W^\top),
\qquad
(b_W)_i=\sum_jW_{ji}-\sum_jW_{ij}.
\]
Order-based candidates are obtained by quantile partitioning,
one-dimensional \(K\)-means on \(s\), and segmentation of the sorted scores.

Combining all views gives a set of initial candidates
\[
\mathcal Z^{(0)}
=
\left\{
(z_r^{\mathrm{out},0},z_r^{\mathrm{in},0})
:
r=1,\ldots,R
\right\}.
\]

\subsection{Local refinement}

Each initial candidate is refined by local sender and receiver updates. The
main refinement optimizes the degree-corrected profile score
\(S_{\mathrm{DC}}\). Suppose the receiver labels are fixed. For a vertex
\(i\), define its outgoing receiver-block vector
\[
h_i^{\mathrm{out}}(b)
=
\sum_j A_{ij}\mathbf 1\{z_j^{\mathrm{in}}=b\}.
\]
If \(i\) moves from sender group \(a\) to \(a'\), only rows \(a\) and \(a'\)
of the block-count matrix \(M\), and the corresponding out-degree totals,
are changed. The local gain is therefore computed by comparing the affected
rows of \(S_{\mathrm{DC}}\) before and after the move. The best positive
gain move is accepted, subject to a minimum group-size constraint. Receiver
updates are analogous and affect only the corresponding columns of \(M\).

TT-SR also generates Bernoulli-refined candidates using the binarized
support \(A^{(0)}_{ij}=\mathbf 1\{A_{ij}>0\}\). In this refinement, a vertex
is assigned to the sender or receiver role that maximizes its Bernoulli
block-likelihood contribution under the current block probabilities
\(\widehat p_{ab}\). These candidates are useful for density-dominated
co-block patterns but are treated conservatively in the final selection
step.

Since the number of sender and receiver assignments is finite, and the
degree-corrected refinement accepts only strictly positive-gain moves, the
refinement terminates at a coordinate-wise local optimum.

\subsection{Two-tier candidate selection}

After refinement, let
\[
\mathcal C=\{c_1,\ldots,c_L\}
\]
be the set of candidates. For each candidate \(c\), TT-SR computes
\[
S_{\mathrm{DC}}(c),
\qquad
S_{\mathrm{B}}(c),
\qquad
S_{\mathrm{Ord}}(c),
\]
and their percentile ranks
\[
r_{\mathrm{DC}}(c),
\qquad
r_{\mathrm{B}}(c),
\qquad
r_{\mathrm{Ord}}(c).
\]
It also computes a robust min--max normalized degree-corrected score
\[
\nu_{\mathrm{DC}}(c)
=
\left[
\frac{
S_{\mathrm{DC}}(c)-q_{0.05}
}{
q_{0.95}-q_{0.05}+\varepsilon
}
\right]_{[0,1]},
\]
where \(q_{0.05}\) and \(q_{0.95}\) are the empirical \(5\%\) and \(95\%\)
quantiles of \(S_{\mathrm{DC}}\) over \(\mathcal C\).

The first tier forms a degree-corrected eligible set
\[
\mathcal C_{\rm gate}
=
\left\{
c\in\mathcal C:
r_{\mathrm{DC}}(c)\ge \rho_{\rm gate}
\quad\text{or}\quad
\nu_{\mathrm{DC}}(c)\ge \gamma_{\rm gate}
\right\}.
\]
Thus, the final output must have strong support under the
degree-corrected sender--receiver criterion. The Bernoulli and order-flow
scores can only influence the choice inside this eligible set.

The second tier selects
\[
c^*
=
\arg\max_{c\in\mathcal C_{\rm gate}}
\left[
w_{\mathrm{DC}}r_{\mathrm{DC}}(c)
+
w_{\mathrm{B}}r_{\mathrm{B}}(c)
+
w_{\mathrm{Ord}}r_{\mathrm{Ord}}(c)
-
\operatorname{penalty}(c)
\right].
\]
A small penalty is applied to non-order Bernoulli-refined candidates:
\[
\operatorname{penalty}(c)
=
\lambda_{\rm pen}
\mathbf 1\{
c\text{ is a non-order Bernoulli-refined candidate}
\}.
\]
Unless otherwise stated, the default synthetic-benchmark selection
parameters are
\[
w_{\mathrm{DC}}=0.45,\qquad
w_{\mathrm{B}}=0.30,\qquad
w_{\mathrm{Ord}}=0.45,
\]
\[
\rho_{\rm gate}=0.90,\qquad
\gamma_{\rm gate}=0.965.
\]
The weights are used only as positive ranking coefficients in the secondary
selection score and are not required to sum to one.  
For non-order Bernoulli-refined candidates, the stricter thresholds
\[
\rho_{\rm gate}^{\rm strict}=0.96,
\qquad
\gamma_{\rm gate}^{\rm strict}=0.985
\]
are used, with \(\lambda_{\rm pen}=0.08\). The real-network diagnostics use
the same two-tier selection rule and the same rank weights, but a slightly
wider degree-corrected gate, as reported in
Section~\ref{subsec:real-networks}.

The final edge type of an observed edge \(i\to j\) is
\[
y_{ij}
=
(z_i^{\mathrm{out}},z_j^{\mathrm{in}}),
\]
and its integer encoding is
\[
\widehat y_{ij}
=
K_{\mathrm{in}}(z_i^{\mathrm{out}}-1)+z_j^{\mathrm{in}}.
\]

\begin{algorithm}[H]
\small
\DontPrintSemicolon
\KwIn{Directed adjacency matrix \(A\); number of roles \(K\); parameter grids}
\KwOut{Sender labels \(z^{\mathrm{out}}\), receiver labels \(z^{\mathrm{in}}\), edge types \(y\), encoded labels \(\widehat y\)}

Construct candidate views from count residuals, stationary-flow residuals,
degree-corrected residuals, and order scores\;

Generate initial candidates
\[
\mathcal Z^{(0)}
=
\{(z_r^{\mathrm{out},0},z_r^{\mathrm{in},0}):r=1,\ldots,R\}.
\]

Set \(\mathcal C\gets\emptyset\)\;

\ForEach{\((z^{\mathrm{out},0},z^{\mathrm{in},0})\in\mathcal Z^{(0)}\)}{
    Add the degree-corrected refined candidate to \(\mathcal C\)\;
    Add the Bernoulli-refined candidate to \(\mathcal C\)\;
    \If{the initialization is order-based}{
        Add the raw order candidate to \(\mathcal C\)\;
    }
}

\ForEach{\(c\in\mathcal C\)}{
    Compute \(S_{\mathrm{DC}}(c)\), \(S_{\mathrm{B}}(c)\), and \(S_{\mathrm{Ord}}(c)\)\;
}

Compute percentile ranks and the robust normalized score \(\nu_{\mathrm{DC}}\)\;

Form the eligible set \(\mathcal C_{\rm gate}\) using the degree-corrected gate\;

Select
\[
c^*
=
\arg\max_{c\in\mathcal C_{\rm gate}}
\left[
w_{\mathrm{DC}}r_{\mathrm{DC}}(c)
+
w_{\mathrm{B}}r_{\mathrm{B}}(c)
+
w_{\mathrm{Ord}}r_{\mathrm{Ord}}(c)
-
\operatorname{penalty}(c)
\right].
\]

Set \((z^{\mathrm{out}},z^{\mathrm{in}})\) to the labels of \(c^*\)\;

\ForEach{observed edge \(i\to j\)}{
    Set \(y_{ij}=(z_i^{\mathrm{out}},z_j^{\mathrm{in}})\) and
    \(\widehat y_{ij}=K_{\mathrm{in}}(z_i^{\mathrm{out}}-1)+z_j^{\mathrm{in}}\)\;
}

\Return \(z^{\mathrm{out}},z^{\mathrm{in}},y,\widehat y\)\;
\caption{\textsc{TT-SR}: Two-Tier Sender--Receiver role detection}
\label{alg:ttsr}
\end{algorithm}

\subsection{Computational complexity}

Let \(n\) be the number of vertices, \(m\) the number of directed edges,
\(K\) the number of roles, \(d\) the number of singular-vector components,
and \(R\) the number of initial candidates. For dense graphs, full SVD costs
\(O(n^3)\). For sparse graphs, truncated or randomized SVD can be applied
through residual matrix-vector products, giving an approximate cost
\[
O(T_{\rm svd}(m+n)d),
\]
where \(T_{\rm svd}\) is the number of Krylov or power iterations
\cite{HalkoMartinssonTropp2011}.

For each candidate, a straightforward full local search for degree-corrected
refinement costs
\[
O(T_{\mathrm{DC}}nK^2),
\]
where \(T_{\mathrm{DC}}\) is the number of refinement sweeps. Bernoulli
refinement has the same order of dependence on \(n\) and \(K\). The final
candidate-selection step is small compared with refinement and costs
\[
O(RK^2+R\log R).
\]

\section{Experiments}
\label{sec:experiments}

The experiments evaluate TT-SR as a node-role-constrained framework for
directed edge community detection. The synthetic benchmarks are organized
by the type of directed structure they are designed to test:
pathway-type communities, sender--receiver co-block communities, and
ordered-flow communities. This organization follows the modeling goal of
TT-SR: an edge community is not treated as an arbitrary cluster of edges,
but as an induced sender--receiver type
\[
i\to j
\quad\mapsto\quad
(z_i^{\mathrm{out}},z_j^{\mathrm{in}}).
\]
Thus, a successful method should recover edge-level communities while also
providing interpretable sender and receiver roles whenever such planted
roles exist.

\subsection{Experimental setup}
\label{subsec:experimental_setup}

\paragraph{Compared methods.}
We compare TT-SR with four baselines:
\[
\text{Bimod-Edge},\qquad
\text{Bimod-RP},\qquad
\text{DiSIM-RP},\qquad
\text{DScore-RP}.
\]
Bimod-Edge is the direct edge-clustering method based on the bimodularity
embedding of Cionca, Chan, and Van De Ville~\cite{Cionca2025Bimodularity}.
It clusters directed edges directly and therefore does not impose that edge
labels are induced by node-level sender and receiver roles. In the synthetic
experiments, we report an oracle-style version of this baseline by selecting
the best result over a small grid of singular-vector dimensions using the
planted edge labels; this makes the comparison conservative in favor of
Bimod-Edge.

The remaining three baselines output sender and receiver roles and then
assign edge labels by the same role-pair rule
\[
i\to j
\quad\mapsto\quad
(z_i^{\mathrm{out}},z_j^{\mathrm{in}}).
\]
Bimod-RP uses the left and right singular-vector embeddings of the directed
modularity matrix
\[
B_A
=
A-\frac{d^{\mathrm{out}}(d^{\mathrm{in}})^\top}{m}.
\]
DiSIM-RP uses a DiSIM-style normalized adjacency matrix
\[
D_{\mathrm{out}}^{-1/2}AD_{\mathrm{in}}^{-1/2}
\]
with the standard zero-inverse convention for zero degrees
\cite{RoheQinYu2016}. DScore-RP uses a D-SCORE-type ratio embedding based on
leading left and right singular-vector coordinates, with numerical
regularization for small denominators~\cite{WangLiangJi2020}. These
role-pair baselines test whether TT-SR improves beyond the role-pair
representation itself.

TT-SR is run with the same algorithmic components in all experiments:
multi-view candidate generation, degree-corrected and Bernoulli local
refinement, order-score candidates, and two-tier candidate selection. The
order-score view is included for all graph families, not only for
ordered-flow graphs. No benchmark-specific modification of the TT-SR
pipeline is made.

For reproducibility, we report the fixed implementation parameters used in
the synthetic benchmark study. For Bimod-Edge, the SVD dimensions are
\[
r\in\{2,3,4,6,8,10\}
\]
at Scale I and
\[
r\in\{2,3,4,5,6,8,10,12\}
\]
at Scales II and III. For TT-SR, the spectral dimensions are
\[
r\in\{2,3,4\},\qquad
r\in\{2,3,4,5\},\qquad
r\in\{2,3,4,5,6\}
\]
at Scales I, II, and III, respectively. The stationary-flow view uses
\[
\lambda\in\{0.02,0.10,0.30,0.50\},
\qquad
\alpha_{\rm flow}=0.05.
\]
The degree-corrected residual view uses
\[
\tau_{\rm out}=\tau_{\rm in}=\xi m/n,
\qquad
\xi\in\{0.25,0.50,1.00\}.
\]
All \(K\)-means steps use \(50\) random initializations. The
degree-corrected and Bernoulli refinements use \(12\) and \(8\) sweeps,
respectively, with Bernoulli smoothing parameter \(1.0\) and minimum group
size \(3\). The order-score view uses ridge parameter \(10^{-3}\), forward
tolerance \(0.10\), \(21\) segmentation grid points, and segmentation
minimum size \(8\). These choices are fixed across all synthetic benchmark
families and scales.

\paragraph{Evaluation protocol.}
All synthetic benchmarks use a three-scale protocol. At Scale I, II, and
III, the number of planted groups is \(K=4,5,6\), respectively. Each
reported value is averaged over \(10\) independent trials with different
random seeds. For methods that require the number of node roles, we use the planted value
of \(K\). For Bimod-Edge, the number of edge clusters is set to the number
of planted edge-label classes, including the background class when present.
In addition, Bimod-Edge is reported in an oracle-style form by selecting
the best result over its tested SVD-dimension grid using the planted edge
labels. In contrast, TT-SR selects its output only by the internal two-tier
criterion and does not use planted labels during candidate selection.

For each observed synthetic edge \(i\to j\), the planted edge label is
defined by the planted sender--receiver block pair of its source and target.
If this block pair belongs to the signal set of the corresponding generator,
the planted edge label is that block pair; otherwise the edge is assigned to
the background category. Edge-community recovery is evaluated by normalized
mutual information and adjusted Rand index,
\[
\mathrm{NMI}
\qquad\text{and}\qquad
\mathrm{ARI},
\]
which measure information agreement and chance-corrected pairwise agreement,
respectively~\cite{StrehlGhosh2002,VinhEppsBailey2010,HubertArabie1985}.
We also compute the Dice score after best matching between predicted and
planted edge classes~\cite{Dice1945}, but omit it from the main tables to
keep the presentation compact.

Since TT-SR also outputs node-level sender and receiver roles, we report
node-role recovery whenever planted sender and receiver labels are
available:
\[
\mathrm{ARI}_{\mathrm{out}}
=
\mathrm{ARI}(z^{\mathrm{out}},z_{\rm true}^{\mathrm{out}}),
\qquad
\mathrm{ARI}_{\mathrm{in}}
=
\mathrm{ARI}(z^{\mathrm{in}},z_{\rm true}^{\mathrm{in}}).
\]

\paragraph{Scale settings.}
The three scale settings are shown in Table~\ref{tab:scale_settings}. The
larger scales increase both the number of vertices and the number of planted
groups. To avoid making larger graphs artificially denser, edge
probabilities are scaled so that expected degrees remain roughly comparable
across scales.

\begin{table}[H]
\centering
\caption{Three scale settings used in the synthetic experiments.}
\label{tab:scale_settings}
\small
\setlength{\tabcolsep}{7pt}
\renewcommand{\arraystretch}{1.15}
\begin{tabular}{lccc}
\toprule
Scale & Number of groups \(K\) & Vertices per group \(n_{\rm per}\) & Total vertices \(n\) \\
\midrule
I   & \(4\) & \(50\) & \(200\) \\
II  & \(5\) & \(70\) & \(350\) \\
III & \(6\) & \(90\) & \(540\) \\
\bottomrule
\end{tabular}
\end{table}

Let \(K_0=4\) and \(n_{{\rm per},0}=50\) denote the Scale-I values. For
probabilities associated with a fixed source--target block pair, we use
\[
p(K,n_{\rm per})
=
p_0\frac{n_{{\rm per},0}}{n_{\rm per}}.
\]
For global noise probabilities, we use
\[
p_{\rm noise}(K,n_{\rm per})
=
p_{{\rm noise},0}
\frac{K_0n_{{\rm per},0}}{Kn_{\rm per}}.
\]
For ordered-flow graphs, forward and backward probabilities are additionally
scaled by
\[
\frac{K_0-1}{K-1},
\]
because each group may connect to several later or earlier groups. Thus,
\[
p_{\rm forward}(K,n_{\rm per})
=
p_{{\rm forward},0}
\frac{n_{{\rm per},0}}{n_{\rm per}}
\frac{K_0-1}{K-1},
\]
and the same rule is used for \(p_{\rm backward}\).

\paragraph{Synthetic benchmark families.}
The benchmark suite contains five representative models grouped into three
directed mechanisms. The Scale-I parameters are summarized in
Table~\ref{tab:representative_settings}; Scale-II and Scale-III parameters
are obtained from the scaling rules above.

\begin{table}[H]
\centering
\caption{Synthetic benchmark families and representative Scale-I settings.}
\label{tab:representative_settings}
\scriptsize
\setlength{\tabcolsep}{3pt}
\renewcommand{\arraystretch}{1.18}
\begin{tabularx}{\textwidth}{l>{\RaggedRight\arraybackslash}X>{\RaggedRight\arraybackslash}X}
\toprule
Model & Scale-I parameters & Main directed signal \\
\midrule
Block-cycle sparse
&
\(p_{\rm self}=0.12,\ p_{\rm con}=0.12,\ p_{\rm noise}=0.005\)
&
Sparse cyclic source--target pathways \\

Oriented blockmodel
&
\(p_{\rm self}=0.18,\ p_{\rm between}=0.22,\ \rho_{\rm forward}=0.90,\ p_{\rm noise}=0.004\)
&
Within-block density with preferred between-block direction \\

Stochastic co-block
&
\(p_{\rm high}=0.22,\ p_{\rm low}=0.006,\ \eta=0.35\)
&
Reusable sender and receiver roles \\

Degree-corrected co-block
&
\(p_{\rm high}=0.18,\ p_{\rm low}=0.004,\ \sigma_{\rm out}=\sigma_{\rm in}=0.8,\ \eta=0.35\)
&
Sender--receiver roles with heterogeneous out- and in-activity \\

Ordered forward-dominant
&
\(p_{\rm within}=0.14,\ p_{\rm forward}=0.22,\ p_{\rm backward}=0.001,\ p_{\rm noise}=0.001\)
&
Global forward flow across ordered groups \\
\bottomrule
\end{tabularx}
\end{table}

The block-cycle model generates edges inside each group with probability
\(p_{\rm self}\), from group \(a\) to group \(a+1\) modulo \(K\) with
probability \(p_{\rm con}\), and between all other block pairs with
probability \(p_{\rm noise}\). The oriented blockmodel generates
within-group edges with probability \(p_{\rm self}\), while between-group
edges are generated with probability \(p_{\rm between}\) and oriented in the
preferred direction with probability \(\rho_{\rm forward}\).

In the stochastic co-block model, each vertex has planted sender and
receiver roles. With probability \(\eta\), the receiver role is reassigned
uniformly among the other \(K-1\) roles. Unless otherwise stated, the active
role pairs are
\[
\mathcal A_{\rm co}=\{(a,a):1\le a\le K\}.
\]
Thus, for an edge \(i\to j\) with
\[
(z_i^{\out},z_j^{\inn})=(a,b),
\]
the Bernoulli probability is \(p_{\rm high}\) if
\((a,b)\in\mathcal A_{\rm co}\), and \(p_{\rm low}\) otherwise. The
degree-corrected version further multiplies this block probability by
independent outgoing and incoming log-normal propensities, with standard
deviations \(\sigma_{\rm out}\) and \(\sigma_{\rm in}\), and clips the
resulting probabilities to \([0,1]\).

In the ordered directed SBM, for a vertex in group \(a\) and a vertex in
group \(b\), the signal probability is \(p_{\rm within}\) if \(a=b\),
\(p_{\rm forward}\) if \(a<b\), and \(p_{\rm backward}\) if \(a>b\). A small
uniform noise probability is added to all admissible block pairs:
\[
p_{ab}
=
p_{\rm noise}
+
\begin{cases}
p_{\rm within}, & a=b,\\
p_{\rm forward}, & a<b,\\
p_{\rm backward}, & a>b.
\end{cases}
\]
All probabilities are chosen so that \(p_{ab}\le 1\).

\subsection{Synthetic benchmark study}
\label{subsec:synthetic_benchmarks}

We evaluate TT-SR on the five synthetic benchmarks summarized in
Table~\ref{tab:representative_settings}, using the three scale settings and
degree-controlled scaling rules described in
Section~\ref{subsec:experimental_setup}. The benchmarks cover three
complementary directed mechanisms: pathway-type communities, co-block
communities, and ordered-flow communities. For each configuration, we report
the mean performance over \(10\) independent trials. Edge-community
recovery is measured by normalized mutual information (NMI) and adjusted
Rand index (ARI).

\begin{table}[H]
\centering
\caption{Synthetic benchmark results. Each entry reports mean NMI/ARI over
\(10\) independent trials. Standard deviations are omitted from the main
table to keep the presentation compact. The abbreviations are BE =
Bimod-Edge, BRP = Bimod-RP, DiRP = DiSIM-RP, and DSRP = DScore-RP.}
\label{tab:synthetic_all_results}
\small
\setlength{\tabcolsep}{3pt}
\renewcommand{\arraystretch}{1.12}
\begin{tabular}{llccccc}
\toprule
Model & Scale & BE & BRP & DiRP & DSRP & TT-SR \\
\midrule
Cycle & I
& 0.9076/0.9118 & 0.6484/0.5480 & 0.6010/0.5006 & 0.6828/0.6104 & \textbf{0.9690/0.9723} \\
Cycle & II
& 0.9100/0.9056 & 0.7311/0.6462 & 0.7046/0.6189 & 0.7424/0.6684 & \textbf{0.9751/0.9770} \\
Cycle & III
& 0.8598/0.8185 & 0.7745/0.6868 & 0.8276/0.7793 & 0.8328/0.7869 & \textbf{0.9783/0.9801} \\
\midrule
Oriented & I
& 0.9735/0.9772 & 0.9379/0.9411 & 0.9368/0.9353 & 0.9469/0.9482 & \textbf{0.9974/0.9976} \\
Oriented & II
& 0.9786/0.9817 & 0.9230/0.9186 & 0.9307/0.9279 & 0.9321/0.9301 & \textbf{0.9974/0.9980} \\
Oriented & III
& 0.9788/0.9813 & 0.9238/0.9162 & 0.9325/0.9262 & 0.9321/0.9252 & \textbf{0.9979/0.9983} \\
\midrule
Co-block & I
& \textbf{1.0000/1.0000} & 0.6776/0.4801 & 0.5069/0.3619 & 0.8858/0.8184 & \textbf{1.0000/1.0000} \\
Co-block & II
& \textbf{1.0000/1.0000} & 0.9612/0.9598 & 0.9473/0.9441 & 0.9652/0.9644 & 0.9998/0.9998 \\
Co-block & III
& \textbf{0.9996/0.9996} & 0.9971/0.9971 & 0.9904/0.9899 & 0.9983/0.9982 & 0.9993/0.9993 \\
\midrule
DC co-block & I
& 0.7315/0.5904 & 0.4719/0.2797 & 0.7584/0.6339 & 0.7877/0.6560 & \textbf{0.9832/0.9858} \\
DC co-block & II
& 0.6410/0.4349 & 0.6940/0.4827 & 0.7124/0.5525 & 0.9024/0.8837 & \textbf{0.9800/0.9817} \\
DC co-block & III
& 0.6355/0.3960 & 0.6222/0.3930 & 0.8766/0.8156 & 0.9357/0.8822 & \textbf{0.9810/0.9827} \\
\midrule
Ordered & I
& 0.3337/0.1767 & 0.3133/0.1406 & 0.4547/0.2996 & 0.3760/0.2222 & \textbf{0.9124/0.9001} \\
Ordered & II
& 0.3218/0.1356 & 0.2732/0.1067 & 0.4496/0.2401 & 0.3528/0.1572 & \textbf{0.9491/0.9405} \\
Ordered & III
& 0.3043/0.1074 & 0.2585/0.0905 & 0.4687/0.2369 & 0.3466/0.1421 & \textbf{0.9615/0.9506} \\
\bottomrule
\end{tabular}
\end{table}

Table~\ref{tab:synthetic_all_results} shows that TT-SR is consistently
strong across the synthetic benchmark suite. On pathway-type graphs, TT-SR
achieves the best NMI and ARI on both sparse block-cycle and oriented
blockmodel benchmarks. The advantage is especially visible on sparse
block-cycle graphs at Scale III, where Bimod-Edge degrades while TT-SR
remains close to perfect recovery.

On stochastic co-block graphs, both Bimod-Edge and TT-SR achieve almost
perfect recovery. This confirms that the role-pair constraint of TT-SR does
not reduce edge-level accuracy in easy co-block settings. The
degree-corrected co-block benchmark is more discriminative. In this setting,
TT-SR substantially outperforms all baselines at every scale under both NMI
and ARI, even when DScore-RP becomes stronger at larger scales. This shows
the value of degree-corrected residuals, local sender--receiver refinement,
and two-tier candidate selection under heterogeneous outgoing and incoming
activity.

On ordered-flow graphs, TT-SR is much more accurate than all compared
baselines. This indicates that the order-score candidates provide useful
information when the dominant structure is a global directed flow rather
than an assortative or local co-block pattern.

For the co-block benchmarks, we also report the node-level sender and
receiver role recovery of TT-SR. This diagnostic is important because TT-SR
does not merely cluster edges; it explains edge communities through
reusable node-level sender and receiver roles.

\begin{table}[H]
\centering
\caption{Node-role recovery of TT-SR on co-block benchmarks. Mean sender
and receiver ARI over \(10\) independent trials.}
\label{tab:synthetic_role_recovery}
\small
\setlength{\tabcolsep}{5pt}
\renewcommand{\arraystretch}{1.12}
\begin{tabular}{llcc}
\toprule
Model & Scale & Sender ARI & Receiver ARI \\
\midrule
Co-block & I   & 1.0000 & 1.0000 \\
Co-block & II  & 1.0000 & 0.9993 \\
Co-block & III & 1.0000 & 0.9987 \\
\midrule
DC co-block & I   & 0.9499 & 0.9221 \\
DC co-block & II  & 0.9260 & 0.9195 \\
DC co-block & III & 0.9258 & 0.9319 \\
\bottomrule
\end{tabular}
\end{table}

Table~\ref{tab:synthetic_role_recovery} shows that the edge-level
performance of TT-SR is supported by meaningful node-level roles. In the
stochastic co-block model, sender and receiver roles are recovered almost
exactly. In the degree-corrected co-block model, role recovery remains high
despite substantial degree heterogeneity, supporting the use of
degree-corrected sender--receiver scoring.

\subsection{Real directed networks}
\label{subsec:real-networks}

We further evaluate the methods on real directed networks whose edges have
explicit source--target semantics. In these networks, an edge \(i\to j\)
does not merely indicate that two vertices are similar; it represents an
asymmetric relation such as sending an email, voting for a candidate, rating
or trusting another user, flying from an origin airport to a destination
airport, or linking from one article to another. Such data are therefore
natural testbeds for sender--receiver community detection.

We use the same five methods as in the synthetic experiments: Bimod-Edge,
Bimod-RP, DiSIM-RP, DScore-RP, and TT-SR. Bimod-Edge represents
unrestricted edge clustering, whereas the other four methods output
role-pair edge labels induced by sender and receiver node roles. Thus, the
real-network comparison again contrasts direct edge clustering with
node-role-constrained sender--receiver clustering.

\paragraph{Bicommunity scoring.}
For a sending set \(S\) and a receiving set \(T\), we use the
bimodularity-style score
\[
q(S,T)
=
\frac{1}{m}
\sum_{i\in S,\ j\in T}
\left(
A_{ij}
-
\frac{d_i^{\mathrm{out}}d_j^{\mathrm{in}}}{m}
\right),
\qquad
m=\sum_{i,j}A_{ij}.
\]
This score measures whether the directed flow from \(S\) to \(T\) is larger
than expected under the directed configuration null model. For a role-pair
method, each pair \((a,b)\) induces a bicommunity
\[
S_a=\{i:z_i^{\mathrm{out}}=a\},
\qquad
T_b=\{j:z_j^{\mathrm{in}}=b\}.
\]
For Bimod-Edge, each edge cluster induces a sending set consisting of the
sources of edges in that cluster and a receiving set consisting of the
targets.

Following the bimodularity interpretation, we rank all detected
bicommunities by their individual scores and report
\[
Q_{\rm bi}^{\max}
=
\max_\ell q_\ell,
\qquad
Q_{\rm bi}^{\rm top3}
=
q_{(1)}+q_{(2)}+q_{(3)},
\]
where \(q_{(1)}\ge q_{(2)}\ge q_{(3)}\) are the three largest individual
bicommunity scores. These quantities are used as structural diagnostics of
the strongest detected sender--receiver pathways. They are not ground-truth
accuracy measures.

\paragraph{Metadata agreement on Email-Eu-core.}
Email-Eu-core~\cite{SNAP,ParanjapeBensonLeskovec2017} is the only real
network in our collection with external node metadata. Each vertex belongs
to one of \(42\) departments. Node labels are
not used during fitting. For evaluation, each observed edge \(i\to j\) is
assigned the ordered metadata label
\[
y_{ij}^{\rm meta}=(\ell_i,\ell_j),
\]
where \(\ell_i\) is the department label of vertex \(i\). For role-pair
methods, the predicted edge label is
\[
\widehat y_{ij}
=
(z_i^{\mathrm{out}},z_j^{\mathrm{in}}).
\]
For Bimod-Edge, the predicted edge label is the edge-cluster label. We
report edge-level NMI and ARI between \(y_{ij}^{\rm meta}\) and
\(\widehat y_{ij}\). For role-pair methods, we also report the better of
the sender-side and receiver-side node NMI against the department labels.

Table~\ref{tab:real_email_eu_metadata} reports the best metadata agreement
over \(K=2,\ldots,12\). This is an oracle diagnostic of the best metadata
alignment achievable within the tested role budgets, not an unsupervised
model-selection result.

\begin{table}[H]
\centering
\caption{Metadata agreement on Email-Eu-core. The reported values are oracle
diagnostics over \(K=2,\ldots,12\). Edge NMI and Edge ARI compare predicted
edge communities with ordered department pairs \((\ell_i,\ell_j)\). Node
NMI is reported only for methods that output reusable sender and receiver
node roles.}
\label{tab:real_email_eu_metadata}
\scriptsize
\setlength{\tabcolsep}{4pt}
\renewcommand{\arraystretch}{1.15}
\begin{tabular}{lcccc}
\toprule
Method & Selected \(K\) & Edge NMI & Edge ARI & Node NMI \\
\midrule
Bimod-Edge & 12 & 0.6006 & 0.1905 & -- \\
Bimod-RP   & 12 & 0.5751 & 0.2210 & 0.3672 \\
DiSIM-RP   & 12 & 0.5386 & 0.2276 & 0.4101 \\
DScore-RP  & 6  & 0.4646 & 0.1562 & 0.3552 \\
TT-SR      & 12 & \textbf{0.6656} & \textbf{0.3849} & \textbf{0.5861} \\
\bottomrule
\end{tabular}
\end{table}

TT-SR achieves the highest edge-level NMI, edge-level ARI, and node-level
NMI on Email-Eu-core. This indicates that the sender--receiver role
constraint is not merely an interpretability device: on this labeled
network, it also yields edge communities that align well with external
metadata.

\paragraph{Leading bicommunities on unlabeled networks.}
For the remaining networks, no ground-truth sender--receiver edge labels
are available. We therefore compare the strength of the leading detected
bicommunities. For each method, we select the best value of
\(Q_{\rm bi}^{\rm top3}\) over the tested \(K\)-grid. For the unlabeled
networks, we use \(K\in\{4,6,8\}\); for Bimod-Edge this corresponds to
\(K^2\) edge clusters, matching the number of possible sender--receiver
role pairs.

Table~\ref{tab:real_unlabeled_combined} combines the basic network
statistics and the leading-bicommunity results. Self-loops are removed, and
repeated directed interactions between the same ordered pair are aggregated
as edge weights. The column \(m_{\rm nz}\) denotes the number of distinct
directed edges after aggregation.

\begin{table}[H]
\centering
\caption{Leading sender--receiver bicommunities on unlabeled real directed
networks. The reported values are the best \(Q_{\rm bi}^{\rm top3}\) over
the tested \(K\)-grid. The column ``Best RP'' is the best value among
Bimod-RP, DiSIM-RP, and DScore-RP.}
\label{tab:real_unlabeled_combined}
\scriptsize
\setlength{\tabcolsep}{5pt}
\renewcommand{\arraystretch}{1.15}
\begin{tabular}{lrrccc}
\toprule
Network & \(n\) & \(m_{\rm nz}\) & Bimod-Edge & Best RP & TT-SR \\
\midrule
CollegeMsg~\cite{SNAP,OpsahlPanzarasa2009,ParanjapeBensonLeskovec2017}
& 1899 & 20296 & 0.1487 & 0.0483 & \textbf{0.3482} \\

Radoslaw-email~\cite{RossiAhmed2015}
& 167 & 5783 & 0.1156 & 0.2393 & \textbf{0.2455} \\

Email-DNC~\cite{RossiAhmed2015}
& 1866 & 5517 & 0.1422 & 0.3336 & \textbf{0.4137} \\

Wiki-Vote~\cite{SNAP,LeskovecHuttenlocherKleinberg2010}
& 7115 & 103689 & 0.2170 & 0.3013 & \textbf{0.3588} \\

Wiki-RfA~\cite{SNAP,LeskovecHuttenlocherKleinberg2010}
& 10493 & 174984 & 0.2093 & 0.2900 & \textbf{0.3885} \\

Bitcoin-Alpha~\cite{SNAP,KumarSpezzanoSubrahmanianFaloutsos2016}
& 3783 & 24186 & 0.1325 & 0.0590 & \textbf{0.3258} \\

Bitcoin-OTC~\cite{SNAP,KumarSpezzanoSubrahmanianFaloutsos2016}
& 5881 & 35592 & 0.1481 & 0.1317 & \textbf{0.3568} \\

Advogato~\cite{RossiAhmed2015}
& 5167 & 47322 & 0.1262 & 0.1048 & \textbf{0.3235} \\

OpenFlights~\cite{OpenFlights}
& 3425 & 37594 & 0.2172 & 0.2147 & \textbf{0.5110} \\

Wikispeedia~\cite{SNAP,WestPineauPrecup2009}
& 4592 & 119772 & 0.0961 & 0.1754 & \textbf{0.3007} \\
\bottomrule
\end{tabular}
\end{table}

Together with Email-Eu-core, the real-data study covers eleven directed
networks: one labeled network for metadata agreement and ten unlabeled
networks for leading-bicommunity diagnostics. TT-SR obtains the largest
\(Q_{\rm bi}^{\rm top3}\) on all ten unlabeled networks. Since these
networks do not have ground-truth sender--receiver edge labels, this result
should be interpreted as structural evidence rather than as an accuracy
comparison. Moreover, because \(Q_{\rm bi}^{\rm top3}\) is itself a
residual-based structural diagnostic, the unlabeled-network results should
not be read as independent ground-truth validation. They indicate that
TT-SR extracts a small number of strong degree-corrected
sender--receiver pathways across communication, voting, trust, route, and
hyperlink networks.

The real-network results are therefore complementary to the synthetic
recovery experiments. On Email-Eu-core, where metadata are available, TT-SR
gives the best oracle metadata agreement within the tested role budgets. On
the unlabeled networks, TT-SR gives the strongest leading-bicommunity
summaries under a common structural diagnostic.

For the real-network diagnostics, Email-Eu-core is evaluated over
\[
K\in\{2,3,\ldots,12\},
\]
while the unlabeled networks are evaluated over
\[
K\in\{4,6,8\}.
\]
The TT-SR candidate grids for real networks are
\[
r\in\{2,3,4\},\qquad
\lambda\in\{0.05,0.30\},\qquad
\xi\in\{0.50,1.00\},
\]
where
\[
\tau_{\rm out}=\tau_{\rm in}=\xi m/n.
\]
The real-network runs use \(20\) \(K\)-means initializations, \(8\)
degree-corrected refinement sweeps, \(5\) Bernoulli refinement sweeps,
Bernoulli smoothing parameter \(1.0\), and minimum group size \(3\). The
stationary-flow computation uses \(\alpha_{\rm flow}=0.05\), maximum
iteration number \(1000\), and tolerance \(10^{-10}\). The order-score view
uses ridge parameter \(10^{-3}\) and forward tolerance \(0.10\).

\paragraph{Parameter-use note.}
The same two-tier selection rule is used in all experiments, but the
candidate-grid size and gate width are chosen separately for synthetic
benchmarks and real-network diagnostics. This distinction is made because
the real networks are larger, more heterogeneous, and do not provide planted
edge labels.

For the synthetic benchmarks, we use
\[
\rho_{\rm gate}=0.90,\qquad
\gamma_{\rm gate}=0.965.
\]
For non-order Bernoulli-refined candidates, the stricter thresholds are
\[
\rho_{\rm gate}^{\rm strict}=0.96,\qquad
\gamma_{\rm gate}^{\rm strict}=0.985.
\]
For the real-network diagnostics, we use the slightly wider gate
\[
\rho_{\rm gate}=0.85,\qquad
\gamma_{\rm gate}=0.94,
\]
with stricter non-order Bernoulli thresholds
\[
\rho_{\rm gate}^{\rm strict}=0.92,\qquad
\gamma_{\rm gate}^{\rm strict}=0.97.
\]
The rank weights are kept fixed in both settings:
\[
w_{\rm DC}=0.45,\qquad
w_{\rm B}=0.30,\qquad
w_{\rm Ord}=0.45,
\]
and the non-order Bernoulli penalty is \(\lambda_{\rm pen}=0.08\).
Thus, the selection mechanism is unchanged, while the gate width is adapted
only between the synthetic recovery setting and the real-network diagnostic
setting, not separately for individual graph families.

\subsection{Summary of experimental results}
\label{subsec:experiment_summary}

The synthetic and real-network experiments provide complementary evidence
for the proposed sender--receiver framework. The synthetic benchmark suite
shows that TT-SR remains accurate across three different directed
mechanisms: pathway-type communities, co-block communities, and ordered-flow
communities. In pathway-type graphs, TT-SR improves over direct
bimodularity edge clustering on sparse block-cycle graphs and remains
nearly perfect on oriented blockmodels. In stochastic co-block graphs,
TT-SR is essentially tied with Bimod-Edge, showing that the role-pair
constraint does not reduce edge-level accuracy in easy co-block settings.
In degree-corrected co-block graphs, TT-SR is substantially more robust
than the baselines, which supports the use of degree-corrected residuals
and local sender--receiver refinement. In ordered-flow graphs, TT-SR clearly
outperforms all tested baselines, showing the usefulness of the order-score
view when the dominant signal is global directed flow.

The real-network experiments support the same conclusion from a different
angle. On Email-Eu-core, TT-SR achieves the strongest alignment with
external department metadata at both the edge and node-role levels. On the
unlabeled networks, where no ground-truth sender--receiver labels are
available, TT-SR consistently extracts higher-scoring leading bicommunities under
the same bimodularity-style diagnostic. Taken together, the experiments
show that TT-SR combines the expressiveness of directed edge clustering
with the interpretability of reusable sender and receiver node roles.

\FloatBarrier

\section{Discussion and Conclusion}

This paper introduced TT-SR, a Two-Tier Sender--Receiver framework for directed edge community detection. 
The central modeling idea is that a directed edge community should not be treated merely as an arbitrary cluster of edges. 
Instead, it can be explained by reusable sender and receiver roles of vertices: each vertex \(i\) has a sender role \(z_i^{\out}\) and a receiver role \(z_i^{\inn}\), and each observed edge \(i\to j\) receives the induced type \((z_i^{\out},z_j^{\inn})\). 
This places TT-SR between one-label vertex clustering and unrestricted edge clustering: it is more expressive than ordinary node clustering, but more interpretable than free edge clustering.

Theoretically, TT-SR is motivated by the stationary-distribution modularity for directed graphs. 
The residual term \(\phi_uP_{uv}-\phi_u\phi_v\) is already directed, but the usual compatibility condition remains one-role. 
TT-SR replaces this one-label condition by a sender--receiver interaction, leading to a two-role modularity viewpoint and a full role-pair representation for directed edge communities. 
The degree-corrected profile score provides the likelihood-based counterpart of this residual viewpoint: its unsmoothed form is a profile likelihood-ratio statistic, and locally it is equivalent to a degree-normalized quadratic sender--receiver residual energy. 
Equivalently, this quadratic energy is the square of the strongest normalized linear sender--receiver residual contrast. 
Algorithmically, TT-SR combines several directed views---count residuals, stationary-flow residuals, degree-corrected residuals, and order-score segmentations---and refines the resulting candidates under common sender--receiver objectives. 
The final two-tier selection rule uses the degree-corrected profile score to define a primary support set, while Bernoulli density and order-flow scores serve only as secondary signals inside this support set.

Across synthetic benchmarks, TT-SR provides accurate edge-level recovery while preserving an interpretable sender--receiver node-role explanation. 
The real-network experiments further show that this representation is useful beyond planted models: it aligns well with Email-Eu-core metadata and extracts prominent directed bicommunity summaries on unlabeled networks.

Several limitations remain. 
First, the synthetic experiments use the planted number of sender and
receiver roles, and the Email-Eu-core metadata scores are oracle diagnostics
over a tested \(K\)-grid rather than the result of a fully unsupervised
model-selection procedure. 
A full model-selection extension should combine the profile score with a
complexity penalty such as BIC \cite{Schwarz1978}, MDL \cite{Rissanen1978},
or held-out likelihood. 
Second, only Email-Eu-core provides external metadata for quantitative
real-network validation; the remaining real-network results are structural
pathway summaries rather than ground-truth accuracy results. 
Third, TT-SR contains several interacting components, including multiple
spectral views, local refinement, and two-tier candidate selection. 
The present experiments evaluate the complete framework against closely
related baselines; a larger-scale component-wise sensitivity study is left
for future work. 
Finally, the current implementation uses several SVD-based initializations
and focuses on static graphs; scalable matrix-free implementations and
dynamic sender--receiver extensions are natural directions for future work.

Together, these results support the sender--receiver role-pair view as a useful and interpretable middle ground between one-label vertex clustering and unrestricted directed edge clustering.

\section*{Declarations}

\subsection*{Competing interests}
The author declares that he has no competing interests.

\subsection*{Funding}
No specific funding was received for the work reported in this manuscript.

\subsection*{Author contributions}
Duy Hieu Do conceived the study, developed the proposed framework and algorithm, designed and implemented the experiments, analyzed the results, and wrote the manuscript.

\subsection*{Data availability statement}
The real-world datasets used in this study are publicly available from the sources cited in the manuscript. The synthetic datasets were generated according to the models and parameter settings described in the manuscript. Additional data supporting the findings of this study are available from the corresponding author upon reasonable request.

\end{document}